\documentclass[12pt,notitlepage]{article}%
\usepackage{amsthm}
\usepackage{setspace}
\usepackage{eurosym}
\usepackage{subfigure}
\usepackage{xcolor}
\usepackage{amssymb}
\usepackage{mathrsfs}
\usepackage{graphicx}
\usepackage[colorlinks,linkcolor=links,citecolor=cites,urlcolor=MyDarkBlue]%
{hyperref}
\usepackage{amsfonts}
\usepackage{amsmath}
\usepackage{amstext}
\usepackage{appendix}
\usepackage{appendix}
\usepackage{mathpazo}
\usepackage{tikz}
\usepackage{verbatim}

\usetikzlibrary{trees}
\usetikzlibrary{matrix}
\usetikzlibrary{decorations.pathreplacing}
\usetikzlibrary{topaths,calc}
\usepackage{natbib}%
\providecommand{\U}[1]{\protect\rule{.1in}{.1in}}
\newtheorem{theorem}{Theorem}

\newtheorem{definition}{Definition}

\newtheorem{lemma}{Lemma}

\newtheorem{proposition}{Proposition}

\numberwithin{equation}{section}

\definecolor{MyDarkBlue}{rgb}{0,0.08,0.45}
\definecolor{cites}{HTML}{324b13}
\definecolor{links}{HTML}{1a663b}
\definecolor{MyLightMagenta}{cmyk}{0.1,0.8,0,0.1}
\hypersetup{
colorlinks,citecolor=blue,filecolor=black,linkcolor=blue,urlcolor=blue
}
\usepackage[top=0.9in,bottom=0.9in,left=0.9in,right=0.9in]{geometry}

\usepackage{setspace}
\usepackage{footmisc}
\usepackage{lipsum}

\begin{document}

\title{Multilateral matching with scale economies\thanks{First online version: Oct, 2023; this version: Feb, 2025.}}
\author{Chao Huang\thanks{Institute for Social and Economic Research, Nanjing Audit University. Email: huangchao916@163.com.}}
\date{}
\maketitle

\begin{abstract}
This paper studies multilateral matching in which agents may negotiate contracts within any coalition. We assume scale economies such that an agent substitutes some existing contracts with new ones only if the latter involve a set of partners that is weakly larger than the original. A weakly setwise stable (or setwise stable) outcome exists and can be found by a Constrained Serial Dictatorship algorithm in markets with scale economies (resp. ordinal scale economies). The scale economies condition applies to an environment in which agents cooperate to achieve targets, such as markets in which countries sign bilateral agreements.
\end{abstract}



\section{Introduction}\label{Sec_intro}

In both business and social contexts, people often engage in multilateral agreements to achieve their goals. For example, firms collaborate to produce goods that require complementary inputs and technologies. Professional services in commercial activities are frequently provided by different organizations. Governments of various countries sign agreements on international affairs. Social activities and relationships can be modeled as a matching market, where an outcome is represented by a set of contracts signed by agents. In these markets, stability is a key equilibrium concept that captures whether an outcome is immune to renegotiations among certain agents. Roughly speaking, an outcome is considered stable if no group of agents can improve their situation by signing new contracts among themselves while possibly dropping some contracts of the original outcome.

It is well-established that stable outcomes exist in two-sided matching markets when there are no complements (\citealp{KC82}, \citealp{RS90},  and \citealp{HM05}). In multilateral markets, where any group of agents can negotiate contracts, \cite{RY20} (henceforth, RY) demonstrated that a stable outcome exists when there are no substitutes. Their conditions are the polar opposite of those in two-sided matching. RY's framework accommodates both transferable and nontransferable utilities, as well as indifferences and externalities. In this paper, we focus on a basic setting of multilateral matching with nontransferable utilities, and without indifferences or externalities. Although this is a special case of RY's framework, our assumptions are widely used in matching theory. When an agent's available set of contracts expands, and she chooses to drop some contracts while signing new ones, we say that she substitutes the dropped contracts with the newly signed ones. In this basic setting, RY's complementarity condition rules out such substitutions. This condition captures the defining feature of complementarities in many real-world multilateral collaborations.

Since real-life matching markets often involve both substitutes and complements, the existence of stable outcomes in the presence of both has been a major challenge in matching theory. A growing body of literature addresses this gap. For instance, \cite{SY06}, \cite{O08}, and \cite{HKNOW13} allowed for within-group substitutes and cross-group complements in different economic settings. \cite{HK10} and \cite{HK19} proposed generalized substitutability conditions in many-to-one matching. \cite{AH18}, \cite{CKK19}, and \cite{NV18,NV19} explored approximate or near-feasible stability. \cite{BK19} and \cite{H23} studied demand types based on agents' choices. 

In this paper, we introduce a condition of \textbf{scale economies} in multilateral matching, defined as follows: An agent substitutes some contracts with new ones only if the newly signed contracts involve a weakly larger set of partners than the dropped contracts. Our assumption extends the complementarity condition by permitting a broad range of substitutes. Scale economies are commonplace in modern economies. In environments where firms collaborate in production or agents work together in business ventures, our assumption applies to cases where having more partners provides advantages. Our framework also allows a group of agents to negotiate over contracts they have signed. For example, in a market for a single type of business, it is natural for a group of agents to negotiate at most one contract. Our condition accommodates this scenario, where contracts signed by the same group of agents are pure substitutes.

We demonstrate that a weakly setwise stable outcome exists in a market with scale economies. Weak setwise stability relaxes the concept of stability by disregarding blocks in which agents make inconsistent decisions. This notion was proposed by \cite{KW09} in the context of many-to-many matching and adopted by \cite{BH21} for multilateral matching. The latter characterized ``venture structures'' that guarantee the existence and efficiency of (weakly) setwise stable outcomes. We define an individually rational outcome as \textbf{constrained efficient} if no other individually rational outcome Pareto dominates it. In a market with scale economies, we prove that an individually rational outcome cannot be weakly setwise blocked if it is constrained efficient. This result implies the existence of a weakly setwise stable outcome in such markets, as there is always at least one individually rational outcome that is constrained efficient. 

We can find an individually rational and constrained efficient outcome using a \textbf{Constrained Serial Dictatorship} (henceforth, CSD) algorithm. In this algorithm, agents are ordered exogenously. Starting with the first agent, each agent sequentially adds contracts into a pool. Specifically, each agent includes contracts signed by herself but not by any agents ordered before her, maximizing her welfare from the contracts in the pool while ensuring that the pool remains part of an individually rational outcome. As the algorithm proceeds, the set of contracts in the pool expands, ultimately resulting in an individually rational and constrained efficient outcome. Consequently, the CSD algorithm produces a weakly setwise stable outcome in a market with scale economies.

We further show that the scale economies condition applies to a class of \textbf{agent-target markets}, where each contract includes a subset of indivisible elementary cooperations among the signatories. In these markets, agents collaborate to achieve their targets, with each target of an agent requiring a specific set of elementary cooperations involving this agent. A concrete example is a market in which countries sign bilateral agreements on trade, cultural exchange, or criminal investigations. Each agreement between two countries specifies detailed collaborations, and each country signs agreements to achieve its targets related to economic development, geopolitics, or international influence. In this environment, a country may substitute one agreement with another if both involve the same partner but would not replace an agreement involving one country with an agreement involving a different country.

We also examine a specific type of market structure where each group of agents negotiates at most one contract—a natural framework for single-sector markets. In this framework, we prove that stability and weak setwise stability coincide, provided substitutions are limited to contracts involving the same group of agents.

Weakly setwise stable outcomes and stable outcomes are not immune to blocks in which agents may coordinate not to choose their best alternatives. Setwise stability, proposed by \cite{S99}, strengthens weak setwise stability to resist such blocks. \cite{EO06} demonstrated that a setwise stable outcome exists in many-to-many matching under a strengthened substitutability condition. We show that a setwise stable outcome exists under a stronger scale economies condition in multilateral matching. We introduce the concept of \textbf{ordinal scale economies} as follows: If an agent holding an individually rational set of contracts becomes better off and maintains individual rationality after dropping some contracts and signing new ones, she continues to improve her position and maintain individual rationality by re-adding any of the dropped contracts that involve partners she does not cooperate with in the newly signed contracts. We prove that, in a market with ordinal scale economies, an individually rational outcome cannot be setwise blocked if and only if it is constrained efficient. Consequently, the CSD algorithm produces a setwise stable outcome in such a market.

The remainder of this paper is organized as follows. Section \ref{Sec_M} introduces the model of multilateral matching with nontransferable utilities. Section \ref{Sec_Scale} presents the scale economies condition, the CSD algorithm, an application, and a special market. Section \ref{Sec_ordinal} proposes the ordinal scale economies condition. Omitted proofs are relegated to the Appendix.

\section{Model\label{Sec_M}}

In this section, we introduce the model of multilateral matching with nontransferable utilities. An alternative setting assumes continuous monetary transfers among agents and often assumes a transferable utility for each agent, see \cite{HK15} and \cite{RY20,RY23}.

There is a finite set $I$ of agents and a finite set $X$ of contracts. Each contract $x\in X$ is signed by a subset of agents $\mathrm{N}(x)\subseteq I$. We assume each contract is signed by at least two agents: $|\mathrm{N}(x)|\geq2$ for each $x\in X$. For each subset of contracts $Y\subseteq X$ and each agent $i\in I$, let $Y_i\equiv\{x\in Y|i\in \mathrm{N}(x)\}$ be the subset of $Y$ in which each contract involves agent $i$, and let $Y_J\equiv\bigcup_{i\in J}Y_i$. For each subset of contracts $Y\subseteq X$, let $\mathrm{N}(Y)\equiv\bigcup_{x\in Y}\mathrm{N}(x)$ be the set of agents who have contracts in $Y$. Each agent $i\in I$ has a strict preference ordering $\succ_i$ over $2^{X_i}$. For any $Y,Z\subseteq X_i$, we write $Y\succeq_iZ$ if $Y\succ_iZ$ or $Y=Z$. Let $\succ_I$ be the preference profile of all agents. A matching market is a tuple $(I,X,N:X\rightrightarrows I,\succ_I)$. An \textbf{outcome} is a set of contracts $Y\subseteq X$. For each agent $i\in N$, let $\mathrm{C}_i:2^X\rightarrow 2^{X_i}$ be agent $i$'s \textbf{choice function} such that $\mathrm{C}_i(Y)\subseteq Y_i$ and $\mathrm{C}_i(Y)\succeq_iZ$ for each $Y\subseteq X$ and $Z\subseteq Y_i$.

\subsection{Stability concepts}\label{Sec_stable}

An outcome $Y\subseteq X$ is \textbf{individually rational} for agent $i\in I$ if $Y_i=\mathrm{C}_i(Y)$. That is, an outcome is individually rational for agent $i$ if agent $i$ does not want to unilaterally drop any contracts. An outcome $Y\subseteq X$ is individually rational if $Y_i=\mathrm{C}_i(Y)$ for all $i\in I$.

\begin{definition}\label{Def_stable}
\normalfont
\begin{description}
\item[(\romannumeral1)] An outcome $Y\subseteq X$ is \textbf{blocked} by a nonempty $Z\subseteq X\setminus Y$ if $Z_i\subseteq \mathrm{C}_i(Y\cup Z)$ for all $i\in \mathrm{N}(Z)$.

\item[(\romannumeral2)] An outcome is \textbf{stable} if it is individually rational and cannot be blocked.
\end{description}
\end{definition}

When an outcome $Y$ is blocked by $Z$, the set of agents $\mathrm{N}(Z)$ improve themselves by signing contracts from $Z$, while possibly dropping some contracts of the outcome $Y$. Moreover, the newly signed contracts are in associated agents' best choices from the newly signed contracts and the original contracts: $Z_i\subseteq \mathrm{C}_i(Y\cup Z)$ for all $i\in \mathrm{N}(Z)$. The following example from \cite{BH21} shows that stability prevents a class of blocks in which agents make inconsistent decisions.
\begin{equation}\label{exam_in1}
\text{Ana}: \{x,y\}\succ\{x\}\succ\emptyset \qquad\qquad\qquad \text{Bob}: \{y\}\succ\{x\}\succ\emptyset
\end{equation}
There is no stable outcome in this market. Specifically, the outcome $\{x\}$ is not stable since both agents want to sign the contract $y$ in the presence of the contract $x$. However, after signing $y$, Bob wishes to drop $x$, while Ana does not. Stability rules out such blocks in which agents make inconsistent decisions about which contracts to drop. Weak setwise stability, proposed by \cite{KW09}, relaxes stability by disregarding such blocks.

\begin{definition}\label{Def_weak}
\normalfont
\begin{description}
\item[(\romannumeral1)] An outcome $Y\subseteq X$ is \textbf{weakly setwise blocked} by a nonempty $Z\subseteq X\setminus Y$ if there exists an outcome $Y^*\subseteq Y\cup Z$ such that $Z_i\subseteq Y^*_i=\mathrm{C}_i(Y\cup Z)$ for all $i\in \mathrm{N}(Z)$.

\item[(\romannumeral2)] An outcome is \textbf{weakly setwise stable} if it is individually rational and cannot be weakly setwise blocked.
\end{description}
\end{definition}

When an outcome $Y$ is weakly setwise blocked by $Z$, the set of agents $\mathrm{N}(Z)$ improve themselves by signing contracts from $Z$, while possibly dropping some contracts of the outcome $Y$. Moreover, (a) the newly signed contracts are in associated agents' best choices from the newly signed contracts and the original contracts: $Z_i\subseteq \mathrm{C}_i(Y\cup Z)$ for all $i\in \mathrm{N}(Z)$; and (b) the agents from $\mathrm{N}(Z)$ make consistent decisions on which contracts to drop. In particular, the part (b) means that the original outcome is brought into a new one, which is implicitly implied in the definition. Notice that $Y_i^*$ with $i\in \mathrm{N}(Z)$ is agent $i$'s best choice from $Y\cup Z$; consequently, the new outcome can be written as $Y^{**}=Y^*\cup V$ where $V=\{x\in Y|i\in \mathrm{N}(x)\text{ implies }i\notin \mathrm{N}(Z)\}$ is the collection of contracts in $Y$ that are not signed by any of $\mathrm{N}(Z)$. In contrast, the block defined in Definition \ref{Def_stable} does not necessarily result in a new outcome, as illustrated in the market (\ref{exam_in1}). Weak setwise stability is weaker than stability since the former prevents blocks in which agents make coordinated decisions, while the latter also prevents blocks in which agents make non-cooperative decisions on which contracts to drop. For instance, the outcome $\{x\}$ in the market (\ref{exam_in1}) is weakly setwise stable but not stable.

Weakly setwise stable outcomes and stable outcomes are not immune to another class of blocks. Consider another example from \cite{BH21}:
\begin{equation}\label{exam_in2}
\text{Ana} : \{x,y\}\succ\{x,z\}\succ\{y\}\succ\{z\}\succ\emptyset, \qquad\qquad \text{Bob} : \{x,z\}\succ\{y\}\succ\emptyset.
\end{equation}
The outcome $\{y\}$ can be blocked as follows: In the presence of the contract $y$, Ana and Bob may sign contracts $x$ and $z$, and drop $y$. Both agents are better off in the new outcome $\{x,z\}$, which is also individually rational for both agents. Stability and weak setwise stability only prevent blocks in which each agent participating in the block chooses her best choice from the original contracts and the newly signed contracts. In the above block, Ana chooses $\{x,z\}$, which is not her best choice from the original contract $y$ and the newly signed contracts $x$ and $z$. Thus, weak setwise stability and stability do not prevent blocks of this form. Setwise stability strengthens weak setwise stability by also preventing blocks of the above form. In other words, setwise stability prevents blocks in which the agents participating in the block deviate to a new outcome that is better and individually rational for them. 

\begin{definition}\label{Def_setwise}
\normalfont
\begin{description}
\item[(\romannumeral1)] An outcome $Y\subseteq X$ is \textbf{setwise blocked} by a nonempty $Z\subseteq X\setminus Y$ if there exists an outcome $Y^*\subseteq Y\cup Z$ such that $Z\subseteq Y^*$, $ Y^*_i\succ_iY_i$, and $Y^*$ is individually rational for all $i\in \mathrm{N}(Z)$.

\item[(\romannumeral2)] An outcome is \textbf{setwise stable} if it is individually rational and cannot be setwise blocked.
\end{description}
\end{definition}

When an outcome $Y$ is setwise blocked by $Z$, the set of agents $\mathrm{N}(Z)$ improve themselves by signing contracts from $Z$ and possibly dropping some contracts of the outcome $Y$. Moreover, (a) the agents from $\mathrm{N}(Z)$ make consistent decisions on which contracts to drop, and thus the original outcome is brought into a new outcome; and (b) the new outcome is better and individually rational for the agents from $\mathrm{N}(Z)$. In particular, in the new outcome, the agents participating in the block may not obtain their best choices from the newly signed contracts and the original contracts, as illustrated by the market (\ref{exam_in2}). Setwise stability is stronger than weak setwise stability since the former prevents blocks of the above form while the latter does not. Setwise stability is independent of stability: Setwise stability prevents blocks of the above form while stability does not, whereas stability prevents blocks in which agents make non-cooperative decisions on which contracts to drop while setwise stability does not. For instance, the outcome $\{x\}$ in the market (\ref{exam_in1}) is setwise stable but not stable, while the outcome $\{y\}$ in the market (\ref{exam_in2}) is stable but not setwise stable.

Let us summarize the three stability concepts. Weak setwise stability is suitable for environments in which agents are bound to certain levels of independence and cooperativeness. In environments where agents collaborate in businesses—if they are neither independent enough to implement a block without reaching consistent decisions nor coordinated enough to give up their best choices—weak setwise stability serves as a reasonable criterion. Since stability and setwise stability are more stringent than weak setwise stability, these two stronger notions also apply to such environments, provided that stable or setwise stable outcomes exist. Moreover, stability can prevent blocks in which agents act independently: when an outcome is blocked, the agents participating in the block do not necessarily make consistent decisions about which contracts to drop. Setwise stability, on the other hand, can prevent blocks in which agents are highly coordinated: when an outcome is setwise blocked, agents may compromise on alternatives that are not their best choices. Therefore, stability applies to environments in which agents act independently, such as the problem of network formation in social media studied by RY. In contrast, setwise stability applies to environments in which agents are highly coordinated, such as markets where firms cooperate in production.

\subsection{Complementary contracts}\label{Sec_com}

Suppose an agent selects certain contracts from a set of available contracts. If additional contracts become available, and the agent drops some existing contracts while signing some of the new ones, we say that the agent has substituted the dropped contracts with the newly signed contracts. However, if any contract chosen by an agent from an available set is always chosen again as the set expands—meaning no substitutions occur—we say that contracts are complementary for this agent. That is, contracts are \textbf{complementary} for agent $i$ if $\mathrm{C}_i(Y)\subseteq \mathrm{C}_i(Y')$ for any $Y$ and $Y'$ with $Y\subset Y'\subseteq X_i$.

An individually rational outcome $Y$ is called \textbf{constrained efficient} if there is no individually rational outcome that Pareto dominates $Y$. That is, there is no individually rational outcome $Y'$ such that $Y'_i\succeq_iY_i$ for all $i\in I$ and $Y'_j\succ_jY_j$ for at least one $j\in I$. It turns out that, if contracts are complementary for all agents, outcomes satisfying different stability concepts coincide with individually rational and constrained efficient outcomes.

\begin{proposition}\label{prop_comp}
\normalfont
Let $Y\subseteq X$ be an outcome in a market in which contracts are complementary for all agents, the following four statements are equivalent.
\begin{description}
\item[(\romannumeral1)] $Y$ is individually rational and constrained efficient.
\item[(\romannumeral2)] $Y$ is stable.
\item[(\romannumeral3)] $Y$ is weakly setwise stable.
\item[(\romannumeral4)] $Y$ is setwise stable.
\end{description}
\end{proposition}

The equivalence between (i) and (ii) follows from RY's result. We can deduce (i)$\Rightarrow$(ii), (i)$\Rightarrow$(iii), and (i)$\Rightarrow$(iv) from the following fact: In a market with complementary contracts, whenever an individually rational outcome $Y$ is blocked by $Z$, weakly setwise blocked by $Z$, or setwise blocked by $Z$, the outcome $Y\cup Z$ is individually rational and Pareto dominates $Y$. The other directions (i)$\Leftarrow$(ii), (i)$\Leftarrow$(iii), and (i)$\Leftarrow$(iv) are due to the following fact: Under the complementarity condition, if an individually rational outcome $Y$ is not constrained efficient, then it is Pareto dominated by an individually rational and constrained efficient outcome $Y'$ that is larger than $Y$: $Y\subset Y'$;\footnote{Note that, under the complementarity condition, if $Y\subseteq X$ is Pareto dominated by an individually rational outcome $Y'$ with $Y\setminus Y'\neq\emptyset$, then $Y\cup Y'$ is individually rational and Pareto dominates both $Y$ and $Y'$.} consequently, the outcome $Y$ is blocked by $Y'\setminus Y$, weakly setwise blocked by $Y'\setminus Y$, and setwise blocked by $Y'\setminus Y$.

\section{Scale economies\label{Sec_Scale}}

In this section, we propose a scale economies condition that is weaker than the complementarity condition and guarantees the existence of a weakly setwise stable outcome.

\subsection{The definition of scale economies}\label{Sec_def}

Let $Y,Y'\subseteq X_i$ be two sets of contracts signed by agent $i$ with $Y\subset Y'$. Agent $i$ chooses the set of contracts $\mathrm{C}_i(Y)$ when her available set of contracts is $Y$. If agent $i$'s available set of contracts expands from $Y$ to $Y'$, and the agent drops contracts from $\mathrm{C}_i(Y)\setminus \mathrm{C}_i(Y')$ and sign contracts from $\mathrm{C}_i(Y')\setminus \mathrm{C}_i(Y)$, we say that agent $i$ uses the newly signed contracts from $\mathrm{C}_i(Y')\setminus \mathrm{C}_i(Y)$ to substitute contracts from $\mathrm{C}_i(Y)\setminus \mathrm{C}_i(Y')$. We assume scale economies in the sense that an agent substitutes contracts with other contracts only if the newly signed contracts involve a weakly larger set of partners.

\begin{definition}
\normalfont
Agent $i$'s preference exhibits \textbf{scale economies} if $\mathrm{N}(\mathrm{C}_i(Y)\setminus \mathrm{C}_i(Y'))\subseteq \mathrm{N}(\mathrm{C}_i(Y')\setminus \mathrm{C}_i(Y))$ for any two sets of contracts $Y$ and $Y'$ with $Y\subset Y'\subseteq X_i$. A market exhibits scale economies if each agent's preference exhibits scale economies.
\end{definition}

The markets (\ref{exam_in1}) and (\ref{exam_in2}) exhibit scale economies since all contracts are signed by the same agents. Recall that, when an outcome is weakly setwise blocked or setwise blocked, the agents participating in the block make consistent decisions about which contracts to drop. As a result, the outcome transitions into a new outcome that is both better and individually rational for the agents involved in the block. This characteristic implies the following lemma.

\begin{lemma}\label{lma_IR}
\normalfont
In a market with scale economies, an individually rational outcome cannot be weakly setwise blocked if it is constrained efficient.
\end{lemma}

To see the role of scale economies, consider an individually rational and constrained efficient outcome for Ana, Bob, and Carol. Suppose this outcome is weakly setwise blocked. The block cannot involve all three agents; otherwise, the three agents would deviate to an individually rational outcome that Pareto dominates the original outcome. Therefore, the block must involve only two agents—say, Ana and Bob. We then consider whether Ana and Bob must drop any contracts signed with Carol as part of the block. If they do not drop any contracts signed with Carol, the block would again result in an individually rational outcome that Pareto dominates the original outcome. This implies that the blocking coalition must drop some contracts signed with Carol. Since the original outcome is individually rational, the dropped contracts must be substituted by new contracts signed by Ana and Bob. However, these new contracts are not signed by a weakly larger set of agents, as Carol is excluded from the new agreements.

Lemma \ref{lma_IR} implies the existence of a weakly setwise stable outcome in a market with scale economies since there always exists an individually rational outcome that is constrained efficient. The converse of Lemma \ref{lma_IR} does not hold: A weakly setwise stable outcome is not necessarily constrained efficient in a market with scale economies. This is shown by the market (\ref{exam_in2}). This market exhibits scale economies since all contracts are signed by the two agents. The outcome $\{y\}$ is weakly setwise stable, but the outcome $\{x,z\}$ is individually rational and Pareto dominates $\{y\}$. Recall that $\{y\}$ is not weakly setwise blocked by $\{x,z\}$ since the contract $z$ does not belong to $\mathrm{C}_{i_1}(\{x,y,z\})=\{x,y\}$.

Now, we introduce two special cases of the scale economies condition.

\begin{definition}\label{Def_special}
\normalfont
\begin{description}
  \item[(i)] \textbf{Contracts signed by different groups are complementary} for agent $i$ if for any $Y,Z\subseteq X_i$ satisfying $\mathrm{C}_i(Y)=Y$ and $Y\cap Z=\emptyset$, $x\in \mathrm{C}_i(Y\cup Z)$ if $x\in \mathrm{C}_i(Y)$ and $\mathrm{N}(z)\neq \mathrm{N}(x)$ for all $z\in Z$.
  \item[(ii)] Agent $i$'s preference exhibits \textbf{single-contract scale economies} if $x\in \mathrm{C}_i(Y)$ and $x\notin \mathrm{C}_i(Y\cup\{y\})$ implies $\mathrm{N}(x)\subseteq \mathrm{N}(y)$ for any $Y\subseteq X_i$, $x\in Y$, and $y\in X_i\setminus Y$. 
\end{description}
\end{definition}

Contracts signed by different groups are complementary for agent $i$ if, given an available set of contracts that is individually rational for agent $i$,   any contract $x$ that belongs to this set will be chosen by agent $i$ when contracts signed by groups different from $\mathrm{N}(x)$ are added into the available set. This is a special case of scale economies since substitutions only arise among contracts signed by the same set of agents.

\begin{lemma}\label{lma_special}
\normalfont
Agent $i$'s preference exhibits scale economies if
\begin{description}
  \item[(i)] contracts signed by different groups are complementary for agent $i$, or
  \item[(ii)] agent $i$'s preference exhibits single-contract scale economies.
\end{description}
\end{lemma}

The next special case requires that if a contract $x$ is dropped by an agent as another contract $y$ becomes available, the contract $y$ must be signed by a weakly larger set of agents than $x$. This condition is stronger than scale economies. For example, consider Ana's preference ordering $\{y,z\}\succ\{x\}\succ\emptyset$. Ana, Bob, and Carol sign the contract $x$; Ana and Bob sign the contract $y$; and the contract $z$ is signed by Ana and Carol. This preference exhibits scale economies: Ana may replace $\{x\}$ with $\{y,z\}$ while the two sets are signed by the same set of agents. On the other hand, Ana's preference fails single-contract scale economies: Ana chooses $x$ from $\{x,z\}$ but drops $x$ as $y$ becomes available, but $y$ is not signed by a weakly larger set of partners than $x$.

\subsection{A Constrained Serial Dictatorship algorithm}\label{Sec_Alg}

Under the complementarity condition, a stable outcome can be found by a novel one-sided Deferred Acceptance algorithm proposed by RY, even when externalities and indifferences are present. In this section, we introduce a \textbf{Constrained Serial Dictatorship} (CSD) algorithm, which generates an individually rational and constrained efficient outcome within our framework. 

The CSD algorithm begins by exogenously determining an ordering of the agents. Then, following this order, the algorithm sequentially introduces one agent to the market at each step. Each agent introduced to the market adds contracts to a pool. These contracts are signed by her but not by the agents who were brought to the market before her. The agent maximizes her welfare from the contracts in the pool provided that the set of contracts in the pool is part of an individually rational outcome. 

Let $\mathcal{A}$ be the collection of all individually rational outcomes, that is,
\begin{equation*}
\mathcal{A}\equiv\{Y\subseteq X|Y_i=\mathrm{C}_i(Y_i) \text{ for all } i\in I\}.
\end{equation*}
Suppose the agents are ordered as $i_1$, $i_2, \ldots, i_{|I|}$, then the algorithm proceeds as follows.

\medskip
\begin{description}
\item[Step 1.] Agent $i_1$ chooses her favorite set of contracts from $X_{i_1}$ among those in which each set is part of an individually rational outcome. That is, agent $i_1$ chooses $Z$ that solves the problem
\begin{equation*}
\max_{\succ_{i_1}}\{Z\subseteq X_{i_1}|Z\subseteq Y \text{ for some } Y\in\mathcal{A}\}
\end{equation*}
Let $Z^1$ be the solution to this problem, and let $Y^1\equiv Z^1$.
\item[Step $k, 2\leq k\leq|I|-1$.] Agent $i_k$ chooses a set of contracts $Z$ from $X_{i_k}\setminus X_{\{i_1,\ldots,i_{k-1}\}}$ such that, among those in which $Y^{k-1}\cup Z$ is part of an individually rational outcome, she prefers $Y_{i_k}^{k-1}\cup Z$ most. That is, agent $i_k$ chooses $Z$ that solves the problem
\begin{equation*}
\max_{\succ_{i_k}}\{Y_{i_k}^{k-1}\cup Z|Z\subseteq X_{i_k}\setminus X_{\{i_1,\ldots,i_{k-1}\}} \text{ and } Y^{k-1}\cup Z\subseteq Y \text{ for some } Y\in\mathcal{A}\}
\end{equation*}
Let $Z^k$ be the solution to this problem, and let $Y^k\equiv Y^{k-1}\cup Z^k$.
\end{description}

The algorithm stops at Step $|I|-1$ and produces $Y^{|I|-1}$. Notice that we do not need the $|I|$th agent to choose contracts. For instance, consider a market with three agents $i_1,i_2$, and $i_3$. The agents $i_1$ and $i_2$ can sign contract $x$; the agents $i_1$ and $i_3$ can sign contract $y$; and the agents $i_2$ and $i_3$ can sign contract $z$. The three agents can sign contract $u$. The agents have the preferences
\begin{equation}\label{exam_alg}
\begin{aligned}
&i_1 : \{x,y,u\}\succ\{x,u\}\succ\emptyset,\\
&i_2 : \{x,z,u\}\succ\{x,u\}\succ\emptyset,\\
&i_3 : \{z,u\}\succ\{y,z\}\succ\emptyset.
\end{aligned}
\end{equation}
The CSD algorithm for the ordering $i_1\rightarrow i_2\rightarrow i_3$ proceeds as follows.

Step 1. The agent $i_1$ adds the contracts $x$ and $u$ to the pool. Notice that she cannot add all three contracts $x$, $y$, and $u$ since $\{x,y,u\}$ is not part of an individually rational outcome: Agent $i_3$ does not sign $y$ and $u$ simultaneously.

Step 2. The agent $i_2$ adds the contract $z$ to the pool. The algorithm terminates and produces the outcome $\{x,z,u\}$.

\begin{lemma}\label{lma_alg}
\normalfont
The CSD algorithm produces an individually rational and constrained efficient outcome.
\end{lemma}

Notice that this lemma does not rely on scale economies. The algorithm produces an individually rational outcome since the agents in the algorithm maintain the contracts in the pool to be part of an individually rational outcome. Suppose there is an individually rational outcome $Y$ that Pareto dominates $Y^{|I|-1}$. According to Step 1 of the algorithm, the outcome $Y$ cannot improve agent $i_1$'s welfare, and thus $Y_{i_1}=Y^{|I|-1}_{i_1}$. Then, Step 2 of the algorithm implies that the outcome $Y$ cannot improve agent $i_2$'s welfare either, and hence we also have $Y_{i_2}=Y^{|I|-1}_{i_2}$. We can finally reach the conclusion $Y=Y^{|I|-1}$, which indicates that $Y^{|I|-1}$ is individually rational and constrained efficient. Therefore, by Lemma \ref{lma_IR}, the outcome $Y^{|I|-1}$ is weakly setwise stable in a market with scale economies.

\begin{theorem}\label{thm_alg}
\normalfont
The CSD algorithm produces a weakly setwise stable outcome in a market with scale economies.
\end{theorem}

Since agent ordering is exogenously determined in the CSD algorithm, the choice of execution sequence inherently creates discriminations among agents. To address this fairness concern in practical applications, practitioners may randomize the agent ordering during implementation.

\subsection{Application: An agent-target market}\label{Sec_app}

In this section, we introduce an \textbf{agent-target market} in which agents cooperate in achieving targets. A contract signed by a group of agents consists of indivisible elementary cooperations among these agents. A target of an agent requires some indivisible elementary cooperations involving this agent. We show that contracts signed by different groups are complementary in this market.

Let $E$ be a set of (indivisible) elementary cooperations in which each cooperation $e\in E$ is implemented by some agents. We abuse the notation to denote by $\mathrm{N}(e)$ the set of agents that implement $e$. For each subset of elementary cooperations $E'\subseteq E$ and each agent $i\in I$, let $E'_i\equiv\{e\in E' | i\in \mathrm{N}(e)\}$ be the subset of $E'$ in which each elementary cooperation involves the agent $i$. Each contract $x\in X$ with $\mathrm{N}(x)=J$ is a nonempty subset of elementary cooperations implemented by members of $J$: $x\in \{E'\subseteq E | E'\neq\emptyset \text{ and } \mathrm{N}(e)=J \text{ for each } e\in E'\}$ if $x\in X$ with $\mathrm{N}(x)=J$. Namely, each contract consists of elementary cooperations implemented by signatories of the contract.

An agent $i\in I$ wants to achieve targets from $T_i$, where each target $t\in T_i$ of agent $i$ requires a nonempty set of elementary cooperations $E^t_i\subseteq E_i$ that involve agent $i$. For any $Y\subseteq X_i$, let $\mathrm{T}_i(Y)\equiv\{t\in T_i | E_i^t\subseteq\bigcup_{x\in Y}x\}$ be the set of targets that can be achieved given the elementary cooperations included in the contracts of $Y$. 

Each agent has a strict preference ordering $\succ_i$ over $2^{X_i}$.  We assume (i) $Y\succ_iY'$ for any $Y,Y'\subseteq X_i$ satisfying $\mathrm{T}_i(Y')\subset \mathrm{T}_i(Y)$; that is, an agent prefers the set of contracts that can achieve a larger set of targets; and (ii) $Y\succ_iY'$ for any $Y,Y'\subseteq X_i$ satisfying $\mathrm{T}_i(Y')=\mathrm{T}_i(Y)$ and $Y\subset Y'$; that is, when two sets of contracts achieve the same set of targets but one of them contains more contracts, the agent prefers the set with less contracts.

\begin{proposition}\label{prop_app}
\normalfont
Contracts signed by different groups are complementary in the agent-target market.
\end{proposition} 

This market allows a contract to be substituted by another one signed by the same agents. For example, an agent may substitute a contract $x=\{e\}$ with another one $x'=\{e,e'\}$ as she manages to achieve a target that requires $e'$. A critical assumption of this market is that an elementary cooperation implemented by agents from $J$ are not enclosed in a contract signed by a larger set of agents $J'\supset J$. Otherwise, a contract signed by $J'$ may be substituted by another contract signed by $J\subset J'$. For instance, an agent would replace a contract $x=\{e,e'\}$ with another one $x'=\{e'\}$ when she achieves a target requiring $e'$ but achieves not target requiring $e$. The scale economies condition fails if $\mathrm{N}(e')= N(x')\subset\mathrm{N}(x)=\mathrm{N}(e)$. This assumption on contents of contracts may fail in real-life markets with multilateral contracts but naturally holds in markets in which agents only sign bilateral contracts (i.e., $\mathrm{N}(x)=2$ for all $x\in X$). An example of this setting is a market where countries establish bilateral agreements on matters such as trade, cultural exchange, and criminal investigations. A bilateral agreement between two countries outlines specific areas of collaboration, enabling a country to pursue its objectives in economic development, geopolitics, or international influence. In this context, a country may choose to substitute one agreement with another when both agreements are made with the same partner country. However, it would not replace an agreement with another one when the two agreements involve different partners.

\subsection{One group negotiating one contract}\label{Sec_one}

One group of agents often signs at most one contract in a market involving a single type of business. In this section, we consider a special case where (i) each group of agents negotiates at most one contract, and (ii) contracts signed by different groups are complementary to the agents. We demonstrate that stability and weak setwise stability coincide in this market, and thus, allowing us to identify a stable outcome using the CSD algorithm. Let 
\begin{equation*}
\mathcal{B}_i\equiv\{Y\subseteq X_i | \mathrm{N}(x)\neq \mathrm{N}(x') \text{ for any } x,x'\in Y\}
\end{equation*}
be the collection of subsets of contracts signed by agent $i$ in which no two contracts are signed by the same set of agents. We assume $Y\succ_i\emptyset$ implies $Y\in\mathcal{B}_i$ for any agent $i\in I$.

By Lemma \ref{lma_special}, agents' preferences exhibit scale economies if contracts signed by different groups are complementary for each agent. We also find that, when an individually rational outcome is blocked in the market described here, agents involved in the blocking coalition do not face non-cooperative decisions about which contracts to terminate. This observation establishes the equivalence between stability and weak setwise stability.

\begin{lemma}\label{lma_stable}
\normalfont
In a market in which one group negotiates at most one contract and contracts signed by different groups are complementary for agents, an outcome is stable if and only if it is weakly setwise stable.
\end{lemma}

To illustrate the role of our assumptions, consider a scenario where Ana and Bob (alongside other agents) block (see Definition \ref{Def_stable}) an individually rational outcome. Suppose Ana wishes to drop a contract $x$ signed with Bob and others. This can only occur if Ana substitutes $x$ with other contracts, as the blocked outcome is individually rational. Under our assumptions, the newly signed contracts must include a contract $y$ involving the same set of agents as $x$. Consequently, Ana and Bob must have mutually agreed to sign $y$, implying Bob also decided to drop $x$ (since each group negotiates at most one contract). Thus, non-cooperative decisions about contract termination never arise during blocking, ensuring stability and weak setwise stability align. 

Theorem \ref{thm_alg}, Lemma \ref{lma_special}, and Lemma \ref{lma_stable} imply the following result.

\begin{proposition}\label{prop_stable}
\normalfont
In a market in which each group negotiates at most one contract and contracts signed by different groups are complementary for agents, the CSD algorithm produces a stable outcome.
\end{proposition}

An example of this market structure is a social network where a pair of agents is connected by at most one link. In a framework with externalities and indifference, RY highlighted complementarities as a defining feature of such markets. In our basic setting, agents may form distinct link types with one another, where each type specifies a unique relationship between them.

\section{Ordinal scale economies}\label{Sec_ordinal}

In this section, we propose an ordinal scale economies condition—stronger than the scale economies condition—that guarantees the existence of a setwise-stable outcome.

A setwise stable outcome does not necessarily exist in a market with scale economies. This is shown by the following market with three agents and four contracts.
\begin{center}
\begin{tikzpicture}[scale=0.6]

  \node (n1) at (0,0) {Bob};
  \node (n2) at (4,0) {Carol};
  \node (n3) at (2,3.4) {Ana};
  \node (n4) at (0.2,1.9) {$x$};
  \node (n5) at (1.7,1.5) {$y$};
  \node (n6) at (3.6,2) {$z$};
  \node (n7) at (2,-0.8) {$u$};
  \draw [very thick](0,0.6) ..controls (0.5,1.8)..(1.3,2.8)
        (0.5,0.6)..controls (1.3,1.6)..(1.7,2.8)
        (3.6,0.6)..controls (3.2,2)..(2.2,2.8)
        (0.9,0)..controls (2,-0.3)..(2.9,0);
\end{tikzpicture}
\end{center}
The agents have the preferences
\begin{equation}\label{exam_nosetwise}
\begin{aligned}
&\text{Ana} : \quad \{x,y,z\}\succ\{x\}\succ\{y,z\}\succ\emptyset,\\
&\text{Bob} : \quad \{x,u\}\succ\{y,u\}\succ\emptyset,\\
&\text{Carol} : \ \ \{z,u\}\succ\emptyset.
\end{aligned}
\end{equation}
Bob chooses the contracts $y$ and $u$ from $\{y,u\}$ but does not choose $y$ when the contract $x$ becomes available. This is the only substitution in this market. The market exhibits scale economies since $x$ and $y$ are both signed by Ana and Bob. This market has a unique nonempty individually rational outcome $\{y,z,u\}$, which is by Lemma \ref{lma_IR} the only weakly setwise stable outcome. However, the outcome $\{y,z,u\}$ is not setwise stable: In the presence of the contracts $y,z,$ and $u$, Ana and Bob may negotiate to sign $x$ and drop $y$ and $z$.

An agent’s preference exhibits ordinal scale economies if the following holds: When an agent becomes better off (while maintaining individual rationality) by replacing existing contracts with new ones, she remains better off and individually rational even after reinstating any dropped contracts involving partners excluded from the newly signed agreements.
\begin{definition}
\normalfont
Agent $i$'s preference exhibits \textbf{ordinal scale economies} if for any sets of contracts $Y,Y'\subseteq X_i$ with $\mathrm{C}_i(Y)=Y$ and $\mathrm{C}_i(Y')=Y'$, $Y'\succ_i Y$ implies $Y'\cup Z=\mathrm{C}_i(Y'\cup Z)$ for any $Z\subseteq Y\setminus Y'$ satisfying $\mathrm{N}(x)\setminus \mathrm{N}(Y'\setminus Y)\neq\emptyset$ for each $x\in Z$. A market exhibits ordinal scale economies if each agent's preference exhibits ordinal scale economies.
\end{definition}
In the above definition, notice that agent $i$ holding contracts from $Y$ becomes better off and maintains individual rationality by replacing contracts from $Y\setminus Y'$ with contracts from $Y'\setminus Y$. The set $Z$ is a subset of the dropped contracts $Y\setminus Y'$, and each contract in $Z$ involves partners she does not cooperative with in the newly signed contracts: $\mathrm{N}(x)\setminus \mathrm{N}(Y'\setminus Y)$ is nonempty for each $x\in Z$. Therefore, the requirement $Y'\cup Z=\mathrm{C}_i(Y'\cup Z)$ means that the agent $i$ continues to become better off and maintain individual rationality by getting back the dropped contracts from $Z$. For instance, Ana's preference in the market (\ref{exam_nosetwise}) fails ordinal scale economies: When she holds contracts $y$ and $z$, she becomes better off and maintains individual rationality by dropping the two contracts and signing contract $x$. Among the dropped contracts, Ana cooperates with Carol in the contract $z$, but she does not cooperate with Carol in the newly signed contract $x$. The ordinal scale economies condition requires Ana to continue to become better off and maintain individual rationality by getting back $z$. However, the set $\{x,z\}$ is not acceptable to Ana. If Ana's preference is changed into \begin{equation}\label{Ana2}
\{x,y,z\}\succ\{x,z\}\succ\{x\}\succ\{y,z\}\succ\emptyset,
\end{equation}
then her preference exhibits ordinal scale economies. 

The ordinal scale economies condition indeed implies scale economies. To illustrate, consider a market with ordinal scale economies where an agent’s available contracts expand. Suppose she replaces some existing contracts with new ones that exclude certain partners. By the ordinal scale economies property, she would still remain better off by reinstating the dropped contracts associated with the excluded partners. This contradicts the optimality of her initial choice among available options, as the reinstated contracts would further improve her position.

\begin{lemma}\label{lma_OSE}
\normalfont
In a market with ordinal scale economies, an individually rational outcome cannot be setwise blocked if and only if it is constrained efficient.
\end{lemma}

The``only if'' direction is straightforward and does not depend on ordinal scale economies: If an individually rational outcome $Y$ is Pareto dominated by another individually rational outcome $Y'$, then $Y$ is setwise blocked by $Y'\setminus Y$. To understand the role of ordinal scale economies in the``if'' direction, consider a constrained efficient and individually rational outcome involving agents Ana, Bob, and Carol. Suppose this outcome is setwise blocked. Using reasoning analogous to Lemma \ref{lma_IR}, we deduce that the blocking coalition involves two agents—e.g., Ana and Bob—and requires dropping contracts with Carol. However, ordinal scale economies dictate that Ana and Bob remain better off and individually rational after reinstating the dropped contracts with Carol. Doing so would transform the allocation into a new individually rational outcome that Pareto dominates the original one—contradicting the original outcome’s constrained efficiency. 

Lemma \ref{lma_alg} and Lemma \ref{lma_OSE} imply the following result.

\begin{theorem}\label{thm_setwise}
\normalfont
The CSD algorithm produces a setwise stable outcome in a market with ordinal scale economies.
\end{theorem}

The ordinal scale economies condition is independent of the complementarity condition: All agents' preferences in the markets (\ref{exam_in1}) and (\ref{exam_in2}) exhibit ordinal scale economies since all contracts are signed by the same agents; but contracts are not complementary for Bob in (\ref{exam_in1}) or for Ana in (\ref{exam_in2}). On the other hand, contracts are complementary in Ana's preference in the market (\ref{exam_nosetwise}), but this preference fails ordinal scale economies.

\section{Appendix}

\subsection{Omitted proofs}

\begin{proof}[Proof of Lemma \ref{lma_IR}]
Suppose an individually rational outcome $Y$ is constrained efficient and weakly setwise blocked by $Z\subseteq X\setminus Y$. Let $Y^*\subseteq Y\cup Z$ be the outcome such that $Z_i\subseteq Y^*_i=\mathrm{C}_i(Y\cup Z)$ for all $i\in \mathrm{N}(Z)$. For each $i\in \mathrm{N}(Z)$, $\emptyset\neq Z_i\subseteq \mathrm{C}_i(Y\cup Z)$ implies $\mathrm{C}_i(Y\cup Z)\neq \mathrm{C}_i(Y)$, and thus, $Y^*_i=\mathrm{C}_i(Y\cup Z)\succ_i \mathrm{C}_i(Y)=Y_i$. We also know that for all $i\in \mathrm{N}(Z)$, $Y^*_i=\mathrm{C}_i(Y\cup Z)\succeq_i\mathrm{C}_i(Y^*)\succeq_iY^*_i$, and thus, $Y^*_i=\mathrm{C}_i(Y^*)$. Hence, $Y^*$ is individually rational for all $i\in \mathrm{N}(Z)$, and all agents from $\mathrm{N}(Z)$ are better off in $Y^*$ than in $Y$.

If $\mathrm{N}(Z)=I$, then $Y^*$ is an individually rational outcome that Pareto dominates $Y$. A contradiction. Hence, $I\setminus \mathrm{N}(Z)$ is nonempty. Let $V=\{x\in Y|i\in \mathrm{N}(x)\text{ implies }i\notin \mathrm{N}(Z)\}$ be the collection of contracts in $Y$ that are not signed by any of $\mathrm{N}(Z)$. Let $Y^{**}=Y^*\cup V$. Recall that when $Y$ is weakly setwise blocked by $Z$, $Y$ is brought into $Y^{**}$. We have $Y^{**}\subseteq Y\cup Z$, $Z_i\subseteq Y^{**}_i=\mathrm{C}_i(Y\cup Z)$ for all $i\in \mathrm{N}(Z)$, and we also know that $Y^{**}$ is individually rational for all $i\in \mathrm{N}(Z)$, and all agents from $\mathrm{N}(Z)$ are better off in $Y^{**}$ than in $Y$.

Suppose for each $j\in I\setminus \mathrm{N}(Z)$, $Y_j\subseteq Y^{**}$, which means that the agents participating in the block do not drop any contracts signed with agents outside $\mathrm{N}(Z)$. Then, since $Y^{**}\subseteq Y\cup Z$, we have $Y^{**}_j=Y_j$ for each $j\in I\setminus \mathrm{N}(Z)$. Since $Y$ is individually rational, we know that $Y^{**}$ is individually rational for all $j\in I\setminus \mathrm{N}(Z)$. Then, since $Y^{**}$ is individually rational for all $i\in \mathrm{N}(Z)$, we know that $Y^{**}$ is individually rational. Since $Y^{**}_i\succ_iY_i$ for each $i\in \mathrm{N}(Z)$, and $Y^{**}_j=Y_j$ for each $j\in I\setminus \mathrm{N}(Z)$, we know that $Y^{**}$ Pareto dominates $Y$.   Hence, $Y^{**}$ is individually rational and Pareto dominates $Y$. A contradiction. We know that there is an agent $j^*\in I\setminus \mathrm{N}(Z)$ with $Y_{j^*}\setminus Y^{**}\neq\emptyset$. Then, since $V\subseteq Y^{**}$, we know that there is a contract $x^*\in Y_{j^*}\setminus Y^{**}$ and an agent $i^*\in \mathrm{N}(Z)$ such that $i^*\in \mathrm{N}(x^*)$. Thus, we have $x^*\in Y_{i^*}=\mathrm{C}_{i^*}(Y)=\mathrm{C}_{i^*}(Y_{i^*})$ and $x^*\notin Y_{i^*}^{**}=\mathrm{C}_{i^*}(Y\cup Z)=\mathrm{C}_{i^*}(Y_{i^*}\cup Z_{i^*})$. Therefore, we have $j^*\in \mathrm{N}( \mathrm{C}_{i^*}(Y_{i^*})\setminus \mathrm{C}_{i^*}(Y_{i^*}\cup Z_{i^*}))$; then since $j^*\in I\setminus \mathrm{N}(Z)$, we have $j^*\notin \mathrm{N}(Z_{i^*})= \mathrm{N}(\mathrm{C}_{i^*}(Y_{i^*}\cup Z_{i^*})\setminus \mathrm{C}_{i^*}(Y_{i^*}))$. This contradicts the scale economies condition.
\end{proof}

\bigskip

\begin{proof}[Proof of Lemma \ref{lma_special}]
(i) For agent $i$, suppose $j\in \mathrm{N}(\mathrm{C}_i(Y)\setminus \mathrm{C}_i(Y'))$ and $j\notin \mathrm{N}(\mathrm{C}_i(Y')\setminus \mathrm{C}_i(Y))$ for some $Y,Y'\subseteq X_i$ with $Y\subset Y'$, then there is $x\in X_j$ satisfying $x\in \mathrm{C}_i(Y)$ and $x\notin \mathrm{C}_i(Y')$. Let $Y''=\mathrm{C}_i(Y)$ and $Z''=\mathrm{C}_i(Y')\setminus \mathrm{C}_i(Y)$. Since $j\notin \mathrm{N}(Z'')$, there is no contract in $Z''$ that is signed by $\mathrm{N}(x)$. Then $x\in \mathrm{C}_i(Y'')=Y''$ and $x\notin \mathrm{C}_i(Y')= \mathrm{C}_i(Y''\cup Z'')$ imply that contracts signed by different groups are not complementary for the agent $i$.

(ii) For any two sets of contracts $Y$ and $Y'$ with $Y\subset Y'\subseteq X_i$, let $k=|\mathrm{C}_i(Y')\setminus \mathrm{C}_i(Y)|$, and $x^1,x^2,\ldots,x^k$ an arbitrary permutation of the contracts from $\mathrm{C}_i(Y')\setminus \mathrm{C}_i(Y)$.\footnote{We omit the trivial case of $\mathrm{C}_i(Y')\setminus \mathrm{C}_i(Y)=\emptyset$, which implies $\mathrm{C}_i(Y')=\mathrm{C}_i(Y)$ and thus $\mathrm{N}(\mathrm{C}_i(Y)\setminus \mathrm{C}_i(Y'))\subseteq \mathrm{N}(\mathrm{C}_i(Y')\setminus \mathrm{C}_i(Y))$.} For each $j\in \{1,2,\ldots,k\}$, let $Z^j=\mathrm{C}_i(Y)\cup\{x^1,x^2,\ldots,x^j\}$; and let $Z^0=\mathrm{C}_i(Y)$. If $x\in \mathrm{C}_i(Z^j)\setminus \mathrm{C}_i(Z^{j+1})$ for some $j\in\{0,\ldots,k-1\} $, we know that $x$ is dropped as $x^{j+1}$ is added into the available set $Z^j$. Now we add the contracts $x^1,x^2,\ldots,x^k$ one by one into $Z^0$. For any $x\in\mathrm{C}_i(Y)\setminus \mathrm{C}_i(Y')$, the contract $x$ must be dropped at some step: $x\in \mathrm{C}_i(Z^j)\setminus \mathrm{C}_i(Z^{j+1})$ for some $j\in\{0,\ldots,k-1\} $ (note that $\mathrm{C}_i(Z^0)=\mathrm{C}_i(Y)$ and $\mathrm{C}_i(Z^k)=\mathrm{C}_i(Y')$).  Then, the single-contract scale economies condition indicates $N(x)\subseteq N(x^j)$. Consequently, we have $\mathrm{N}(\mathrm{C}_i(Y)\setminus \mathrm{C}_i(Y'))\subseteq \mathrm{N}(\mathrm{C}_i(Y')\setminus \mathrm{C}_i(Y))$.
\end{proof}

\bigskip

\begin{proof}[Proof of Proposition \ref{prop_app}]
Let $Y\subseteq X_i$ be an individually rational set of contracts signed by agent $i$: $Y=\mathrm{C}_i(Y)$. For any contract $x\in Y$, we know that $x$ is not redundant: there exists $e'\in x$ such that $e'\notin\bigcup_{y\in Y: y\neq x}y$ and $e'\in E^{t'}_i$ for some $t'\in \mathrm{T}_i(Y)$. Otherwise, there is $Y'\subseteq Y\setminus\{x\}$ with $\mathrm{T}_i(Y')=\mathrm{T}_i(Y)$; then the assumption (ii) implies $Y'\succ_i Y$, which contradicts $Y=\mathrm{C}_i(Y)$. Suppose $x\notin\mathrm{C}_i(Y\cup Z)$ for some $Z\subseteq X_i$ with $\mathrm{N}(x)\neq\mathrm{N}(z)$ for all $z\in Z$. Since $e'\notin\bigcup_{y\in Y: y\neq x}y$ with $\mathrm{N}(e')=\mathrm{N}(x)\neq\mathrm{N}(z)$ for all $z\in Z$, we have $e'\notin\bigcup_{y\in\mathrm{C}_i(Y\cup Z)}y$, and thus $t'\notin\mathrm{T}_i(\mathrm{C}_i(Y\cup Z))$. Hence, there is $Y''\subseteq Y\cup Z$ such that $t'\in\mathrm{T}_i(Y'')$ and $\mathrm{T}_i(\mathrm{C}_i(Y\cup Z))\subset\mathrm{T}_i(Y'')$. Then, according to the assumption (i), we have $Y''\succ_i\mathrm{C}_i(Y\cup Z)$ which contradicts $Y''\subseteq Y\cup Z$.
\end{proof}

\bigskip

\begin{proof}[Proof of Lemma \ref{lma_stable}]
The ``only if'' part is trivial. The ``if'' part is equivalent to the statement that an individually rational outcome can be weakly setwise blocked if it can be blocked. Let $Y$ be an individually rational outcome. Suppose $Y$ is blocked by $Z\subseteq X\setminus Y$, and let $Y^*=\bigcup_{i\in \mathrm{N}(Z)}\mathrm{C}_i(Y\cup Z)$. Since $Z_i\subseteq \mathrm{C}_i(Y\cup Z)$ for all $i\in \mathrm{N}(Z)$, we have $Z_i\subseteq Y^*_i$ for all $i\in \mathrm{N}(Z)$. We want to prove that $Y$ is weakly setwise blocked by $Z$, then it suffices to show $Y_i^*=\mathrm{C}_i(Y\cup Z)$ for all $i\in \mathrm{N}(Z)$, which holds if there is no $x\in \mathrm{C}_i(Y\cup Z)$ for some $i\in \mathrm{N}(Z)$ such that $x\notin \mathrm{C}_j(Y\cup Z)$ for some $j\in \mathrm{N}(Z)\cap \mathrm{N}(x)$. Suppose there exists such a contract $x$ satisfying $x\in \mathrm{C}_i(Y\cup Z)$ for some $i\in \mathrm{N}(Z)$ and $x\notin \mathrm{C}_j(Y\cup Z)$ for some $j\in \mathrm{N}(Z)\cap \mathrm{N}(x)$. Since $x\notin \mathrm{C}_j(Y\cup Z)$, we have $x\notin Z$, otherwise $Z_j\subseteq \mathrm{C}_j(Y\cup Z)$ does not hold. Thus, we have $x\in Y$. The individual rationalities of $Y$ indicate $x\in \mathrm{C}_j(Y_j)$. Then, since contracts singed by different groups are complementary, $x\notin \mathrm{C}_j(Y_j\cup Z_j)$ implies that there is $x'\in Z_j$ with $\mathrm{N}(x)=\mathrm{N}(x')$. Notice that $i\in \mathrm{N}(x)=\mathrm{N}(x')$, $x\in \mathrm{C}_i(Y\cup Z)$, and $x'\in Z\subseteq \mathrm{C}_i(Y\cup Z)$. Thus, agent $i$ chooses both $x$ and $x'$ from $Y\cup Z$. This contradicts the assumption that each group negotiates at most one contract.
\end{proof}

\bigskip

\begin{proof}[Proof of Lemma \ref{lma_OSE}]

The ``only if'' part is straightforward. The ``if'' part: Suppose an individually rational outcome $Y$ is constrained efficient and setwise blocked by $Z\subseteq X\setminus Y$. Let $Y^*\subseteq Y\cup Z$ be the outcome such that $Z\subseteq Y^*$, $ Y^*_i\succ_iY_i$, and $Y^*$ is individually rational for all $i\in \mathrm{N}(Z)$. If $\mathrm{N}(Z)=I$, then $Y^*$ is an individually rational outcome that Pareto dominates $Y$. A contradiction. Hence, $I\setminus \mathrm{N}(Z)$ is nonempty. Let $V=\{x\in Y|i\in \mathrm{N}(x)\text{ implies }i\notin \mathrm{N}(Z)\}$ be the collection of contracts in $Y$ that are not signed by any of $\mathrm{N}(Z)$. Let $Y^{**}=Y^*\cup V$. Recall that when $Y$ is setwise blocked by $Z$, $Y$ is brought into $Y^{**}$. We have $Y^{**}\subseteq Y\cup Z$, $ Y^{**}_i\succ_iY_i$, and $Y^{**}$ is individually rational for all $i\in \mathrm{N}(Z)$.

Suppose for each $j\in I\setminus \mathrm{N}(Z)$, $Y_j\subseteq Y^{**}$, which means that the agents participating in the block do not drop any contracts signed with agents outside $\mathrm{N}(Z)$. Then, since $Y^{**}\subseteq Y\cup Z$, we have $Y^{**}_j=Y_j$ for each $j\in I\setminus \mathrm{N}(Z)$. Since $Y$ is individually rational, we know that $Y^{**}$ is individually rational for all $j\in I\setminus \mathrm{N}(Z)$. Then, since $Y^{**}$ is individually rational for all $i\in \mathrm{N}(Z)$, we know that $Y^{**}$ is individually rational. Since $Y^{**}_i\succ_iY_i$ for each $i\in \mathrm{N}(Z)$, and $Y^{**}_j=Y_j$ for each $j\in I\setminus \mathrm{N}(Z)$, we know that $Y^{**}$ Pareto dominates $Y$. Hence, $Y^{**}$ is individually rational and Pareto dominates $Y$. A contradiction.

We know that $(Y\setminus Y^{**})_{I\setminus \mathrm{N}(Z)}\neq\emptyset$. Let $Z'=(Y\setminus Y^{**})_{I\setminus \mathrm{N}(Z)}$. For each $x\in Z'$, since $Y^{**}\setminus Y=Z$, we have $\mathrm{N}(x)\setminus \mathrm{N}(Y^{**}\setminus Y)=\mathrm{N}(x)\setminus \mathrm{N}(Z)\neq\emptyset$. Then, according to the ordinal scale economies condition, since $Y^{**}_i\succ_iY_i$ for each $i\in \mathrm{N}(Z)$, we know that for each $i\in \mathrm{N}(Z)$, $\mathrm{C}_i(Y^{**}_i\cup Z'_i)=Y^{**}_i\cup Z'_i$. Then, since $Y^{**}_j\cup Z'_j=Y_j$ for each $j\in I\setminus \mathrm{N}(Z)$, we know that $Y^{**}\cup Z'$ is an individually rational outcome that Pareto dominates $Y$. This also contradicts that $Y$ is individually rational and constrained efficient.
\end{proof}


\bigskip

\begin{thebibliography}{99999999999999999999999999999999999999999}

\bibitem[Azevedo and Hatfield(2018)]{AH18}{\small Azevedo, E.M., Hatfield, J.M., 2018. Existence of equilibrium in large matching markets with complementarities. \emph{working paper}.}
    
\bibitem[Baldwin and Klemperer(2019)]{BK19}{\small Baldwin, Elizabeth and Paul Klemperer (2019), ``Understanding Preferences: Demand Types'', and the Existence of Equilibrium with Indivisibilities. \emph{Econometrica}, 87, 867-932. }

\bibitem[Bando and Hirai(2021)]{BH21}{\small Bando, K., Hirai, T., 2021. Stability and venture structures in multilateral matching. \emph{Journal of Economic Theory}, 196, 105292.}  
  
\bibitem[Che et al.(2019)]{CKK19}{\small Che, Y-K., Kim, J., Kojima, F., 2019. Stable matching in large economies. \textit{Econometrica}, 87(1), 65-110.}
    
\bibitem[Echenique and Oviedo(2006)]{EO06}{\small Echenique, F., Oviedo, J., 2006. A theory of stability in many-to-many matching. \textit{Theoretical Economics}, 1, 233-273.}

\bibitem[Gale and Shapley(1962)]{GS62}{\small Gale, D., Shapley, L.S., 1962. College admissions and the stability of marriage. \textit{American Mathematical Monthly}, 69, 9-15.}
  
\bibitem[Hatfield and Kojima(2010)]{HK10}{\small Hatfield, John W. and Fuhito Kojima (2010), Substitutes and stability for matching with contracts. \textit{Journal of Economic Theory}, 145(5), 1704-1723.}
    
\bibitem[Hatfield and Kominers(2015)]{HK15}{\small Hatfield, J.W., Kominers, S.D., 2015. Multilateral matching. \textit{Journal of Economic Theory}, 156, 175-206.}

\bibitem[Hatfield and Kominers(2019)]{HK19}{\small Hatfield, John W. and Scott D. Kominers (2019), Hidden substitutes. \emph{Working paper}.}

\bibitem[Hatfield et al.(2013)]{HKNOW13}{\small Hatfield, J. W., Kominers, S. D.,  Nichifor, A., Ostrovsky, M., Westkamp, A., 2013. Stability and competitive equilibrium in trading networks. \textit{Journal of Political Economy}, 121(5), 966-1005.}

\bibitem[Hatfield and Milgrom(2005)]{HM05}{\small Hatfield, J.W., Milgrom, P.R., 2005. Matching with contracts. \textit{American Economic Review}, 95, 913-935.}
    
\bibitem[Huang(2023)]{H23}{\small Huang, C., 2023. Stable matching: An integer programming approach. \emph{Theoretical Economics}, 18, 37-63.}
    
\bibitem[Klaus and Walzl(2009)]{KW09}{\small Klaus, B., Walzl, M., 2009. Stable many-to-many matchings with contracts. \textit{Journal of Mathematic Economics}, 45, 422-434.}

\bibitem[Kelso and Crawford(1982)]{KC82}{\small Kelso, A. S., Crawford, V.P., 1982. Job matching, coalition formation and gross substitutes. \textit{Econometrica}, 50, 1483-1504. }

\bibitem[Nguyen and Vohra(2018)]{NV18}{\small Nguyen, T., Vohra, R., 2018. Near-Feasible Stable Matchings with Couples. \textit{American Economic Review}, 108(11), 3154-3169.}
    
\bibitem[Nguyen and Vohra(2019)]{NV19}{\small Nguyen, T., Vohra, R., 2019. Stable matching with proportionality constraints. \textit{Operations Research}, 67(6), 1503-1519.}

\bibitem[Ostrovsky(2008)]{O08}{\small Ostrovsky, M., 2008. Stability in supply chain networks. \textit{American Economic Review}, 98, 897-923. }

\bibitem[Rostek and Yoder(2020)]{RY20}{\small Rostek, M., Yoder, N., 2020. Matching with complementary contracts. \textit{Econometrica}, 88(5), 1793-1824.}

\bibitem[Rostek and Yoder(2023)]{RY23}{\small Rostek, M., Yoder, N., 2023. Complementarity in matching markets and exchange economies. \textit{Working paper}.}

\bibitem[Roth and Sotomayor(1990)]{RS90}{\small Roth, A.E., Sotomayor, M., 1990. Two-sided matching: A study in game-theoretic modelling and analysis. Econometric Society Monographs No. 18, Cambridge University Press, Cambridge England.}

\bibitem[Sotomayor(1999)]{S99}{\small Sotomayor, M.A.O., 1999. Three remarks on the many-to-many stable matching problem. \textit{Mathematical Social Sciences}, 38, 55-70.}

\bibitem[Sun and Yang(2006)]{SY06}{\small Sun, N., Yang, Z., 2006. Equilibria and indivisibilities: Gross substitutes and complements. \textit{Econometrica}, 74, 1385-1402.}

\end{thebibliography}
\end{document}